\documentclass[letterpaper, 10 pt, conference]{ieeeconf}  % Comment this line out
\IEEEoverridecommandlockouts
\usepackage{cite}
\usepackage{amsmath,amssymb,amsfonts}
\usepackage{algorithmic}
\usepackage{graphicx}
\usepackage{textcomp}
\usepackage{xcolor}
\def\BibTeX{{\rm B\kern-.05em{\sc i\kern-.025em b}\kern-.08em
    T\kern-.1667em\lower.7ex\hbox{E}\kern-.125emX}}
    
\usepackage{hyperref}
\usepackage{helvet}
\usepackage{courier}
\usepackage{amssymb} %maths
\usepackage{amsmath} %maths
\usepackage[ruled,vlined,linesnumbered]{algorithm2e}
\usepackage{dsfont}
\DeclareMathOperator*{\argmax}{argmax}   % Jan Hlavacek
\usepackage{booktabs} % nicer tables :-)
\usepackage{multirow}   % afor multi-row tables
\usepackage{array, makecell} %
\usepackage{graphics}
\usepackage{amsthm}
\usepackage{balance}
\usepackage[switch]{lineno}
\usepackage{tcolorbox}

\DeclareMathOperator*{\argmin}{\arg\!\min}

\newtheoremstyle{boldStyle}% name of the style to be used
  {\topsep}% measure of space to leave above the theorem. E.g.: 3pt
  {\topsep}% measure of space to leave below the theorem. E.g.: 3pt
  {\itshape}% name of font to use in the body of the theorem
  {0pt}% measure of space to indent
  {\bfseries}% name of head font
  {.}% punctuation between head and body
  { }% space after theorem head; " " = normal interword space
  {\thmname{#1}\thmnumber{ #2}\thmnote{ (#3)}}

\newtheoremstyle{italicStyle}% name of the style to be used
  {\topsep}% measure of space to leave above the theorem. E.g.: 3pt
  {\topsep}% measure of space to leave below the theorem. E.g.: 3pt
  {}% name of font to use in the body of the theorem
  {0pt}% measure of space to indent
  {\bfseries}% name of head font
  {.}% punctuation between head and body
  { }% space after theorem head; " " = normal interword space
  {\thmname{#1}\thmnumber{ #2}\thmnote{ (#3)}}

\theoremstyle{boldStyle}
\newtheorem{theorem}{Theorem}
\newtheorem{lemma}{Lemma}

\theoremstyle{italicStyle}
\newtheorem{assumption}{Assumption}

\def\sq{\mathbin{{\strut\rule{1.25ex}{1.25ex}}}}
\renewenvironment{proof}{{\textbf{Proof:}}}{\hfill$\sq$}

\definecolor{greenColor}{rgb}{0.6,0.8,0.0}
\definecolor{ligthGray}{rgb}{0.95,0.95,0.95}

\newcommand{\be}{b_\mathcal{E}}
\newcommand{\Be}{\mathcal{B}_\mathcal{E}}
\newcommand{\G}{\mathcal{G}}    
\newcommand{\T}{\mathcal{T}}    
    
\def\mathcolor#1#{\@mathcolor{#1}}
\def\@mathcolor#1#2#3{%
  \protect\leavevmode
  \begingroup\color#1{#2}#3\endgroup
}

\title{\LARGE \bf
Time-Optimal  Navigation in Uncertain Environments \\ with High-Level Specifications}
\author{Ugo Rosolia, Mohamadreza Ahmadi, Richard M. Murray, and Aaron D. Ames  % <-this % stops a space

\thanks{The authors are with the California Institute of Technology, 1200 E. California Blvd., Pasadena, CA 91125.
e-mails: \tt\scriptsize{\{urosolia, mrahmadi, murray, ames\}@caltech.edu.}}
}%

\begin{document}
\maketitle

% \vspace{-0.4cm}

\begin{abstract}
Mixed observable Markov decision processes (MOMDPs) are a modeling framework for autonomous systems described by both fully and partially observable states.
In this work, we study the problem of synthesizing a control policy for MOMDPs that minimizes the expected time to complete the control task while satisfying syntactically co-safe Linear Temporal Logic (scLTL) specifications. First, we present an exact dynamic programming update to compute the value function. Afterwards, we propose a point-based approximation, which allows us to compute a lower bound of the closed-loop probability of satisfying the specifications. The effectiveness of the proposed approach and comparisons with standard strategies are shown on high-fidelity  navigation tasks with partially observable static obstacles. 
% Finally, we show that the proposed strategy, which leverages mixed-observability, is suitable for high-level decision making of autonomous systems operating in partially known environments.
\end{abstract}

% \begin{IEEEkeywords}
% mixed observability, decision making, partial observation, specifications, time-optimal
% \end{IEEEkeywords}

\vspace{-0.2cm}
\section{Introduction}
Autonomous systems take actions based on observations of the environment surrounding them. When the environment includes both fully observable and partially observable regions, mixed observable Markov decision processes (MOMDPs) can be used as a framework for decision making under uncertainty~\cite{ong2010planning}. In MOMDPs, the state space is partitioned into fully observable and partially observable states. Decisions are taken based on the fully observable states and the \textit{belief} representing a probability distribution over the partially observable states. 
Compared to partially observable Markov decision processes (POMDPs), which maintain a belief for all possible states~\cite{sondik1978optimal}, MOMDPs allow us to reduce the computational complexity of the policy synthesis process when both partial and full state observations are available~\cite{ong2010planning}.

In POMDPs and MOMDPs, the control objective is usually expressed as a reward maximization problem~\cite{sondik1978optimal}. However, reward maximization alone cannot fully encode the desired high-level objectives. Thus, researchers have focused on constrained POMDPs (CPOMDPs), where the synthesis goal is to compute a policy that maximizes the expected reward, while satisfying expected constraints. This problem was first studied in~\cite{isom2008piecewise}, where the authors presented an exact dynamic programming update to compute the optimal deterministic policy. The computational complexity of solving this problem is double exponential in the time horizon. But, the optimal solution can be approximated in polynomial time using point-based~\cite{kim2011point} and finite-state~\cite{poupart2015approximate} approximations.
% have been proposed to reduce the computational complexity associated with policy synthesis. 

Whenever temporal properties of the system are of interest, control objectives can also be expressed using Linear Temporal Logic (LTL) formulas~\cite{pnueli1977temporal}. The \textit{qualitative problem} of synthesizing a policy, which guarantees satisfaction of LTL formulas for POMDPs, is undecidable when searching over the set of feedback policies and EXPTIME-complete when designing finite-state controllers~\cite{chatterjee2015qualitative,chatterjee2016decidable,junges2018finite,junges2020enforcing,ahmadi2020stochastic}. 
When the system is uncertain, it may be impossible to design a policy that guarantees satisfaction of the specifications for all possible uncertainty realizations. In this case, it is desirable to solve the \textit{quantitative problem}, where the objective is to synthesize a policy that maximizes the probability of satisfying LTL specifications. The solution to this quantitative problem can be approximated by discretizing the belief space~\cite{nilsson2018toward}, leveraging finite state controllers~\cite{ahmadi2020stochastic} or using point-based and simulation-based strategies~\cite{bouton2020point,haesaert2019temporal,vasile2016control,haesaert2018temporal,wang2018bounded, lesser2016approximate}.
The optimal solution to quantitative problems is usually not unique; instead there exists a set of optimal control policies~\cite{ding2014optimal}.
% that maximize the probability of satisfying the specifications~\cite{ding2014optimal}. 
For this reason, it is often preferable to compute an optimal policy, which maximizes an expected reward while satisfying LTL specifications~\cite{karaman2008optimal,wolff2016optimal,ding2014optimal}.

In this work, we consider \textit{time-optimal quantitative} problems, where the goal is to minimize the expected time to complete the task while satisfying syntactically co-safe LTL (scLTL) specifications. These problems have been studied for deterministic systems in~\cite{karaman2008optimal, wolff2016optimal} and in~\cite{ding2014optimal, teichteil2012path, teichteil2012stochastic,trevizan2017efficient,lacerda2019probabilistic,kolobov2012theory} for Markov decision processes.
To the knowledge of the authors, this is the first work that studies time-optimal quantitative problems for mixed observable Markov decision processes.
Our contribution is threefold. First, we present a dynamic programming update to compute the value function associated with the time-optimal quantitative problem. Second, we propose a point-based strategy to approximate the optimal value function and we show that our approach maximizes a lower bound of the closed-loop probability of satisfying the specifications. Finally, we compare our method with standard time-optimal and quantitative~policies. 
We show that the proposed strategy allows us to minimize the expected time to complete the task without compromising the probability of satisfying the specifications.
% Therefore, the proposed strategy can be used as a high-level decision maker for navigation tasks with mixed observability, as shown in the example section.

\noindent
\textit{Notation: } For a vector $\alpha \in \mathbb{R}^n$ and an integer $s \in \{1,\ldots, n\}$ we use $\alpha(s)$ to denote the $s$th component of the vector $\alpha$ and $\alpha^\top$ to indicate its transpose. For a function $V : \mathbb{R}^n \rightarrow \mathbb{R}$, $V(\alpha)$ denotes the value of the function $V$ at $\alpha$. Throughout the paper, we will use capital letters to indicate functions and lower letters to indicate vectors.
Given two sets $\mathcal{A}$ and $\mathcal{B}$, 
the set minus operation is denoted as $\mathcal{A}\setminus\mathcal{B}$ and the Cartesian product as $\mathcal{A}\times\mathcal{B}$.  Furthermore, we define the indicator function $\mathds{1}_\mathcal{A}(x) = 1$ if $x\in \mathcal{A}$ and $\mathds{1}_\mathcal{A}(x) = 0$ otherwise. The vectors of ones is written as $1_{n} \in \mathbb{R}^{n}$ and zeros as $0_{n} \in \mathbb{R}^{n}$. Finally, given two sets of vectors $\Gamma=\{ \gamma_i | \forall i\in\{1, \ldots, n_\gamma \}\}$ and $\Lambda=\{ \lambda_j  | \forall j\in\{1, \ldots, n_\lambda \}\}$ we denote $\Gamma \oplus\Lambda = \{ \gamma_i +\lambda_j|  \forall i\in\{1, \ldots, n_\gamma \},  \forall j\in\{1, \ldots, n_\lambda \} \}$.

\section{Background}

In this section, we introduce some definitions and assumptions used in the sequel.

\noindent
\textbf{Mixed Observable Markov Decision Process}\\
A MOMDP provides a sequential  decision-making formalism for high-level planning under mixed full and partial observations~\cite{ong2010planning}. 
% At every time step, the agents take actions and receive observations. 
% These observations are shared via (noise and delay free) communication and the agents decide in a centralized framework.
% \vspace{0.2cm}
%\begin{definition}\label{defi:mpomdp}
More formally, a MOMDP $\mathcal{M}$ is a tuple $\left( \mathcal{S}, \mathcal{E}, \mathcal{A},\mathcal{Z}, T_s, T_e, O \right)$, where 
\begin{itemize}

	\item $\mathcal{S}=\{1,\ldots,|\mathcal{S}|\}$ is a set of fully observable states;
	
	\item $\mathcal{E}=\{1,\ldots,|\mathcal{E}|\}$ is a set of partially observable states;
    
    \item $\mathcal{A}=\{1,\ldots,|\mathcal{A}|\}$ is a set of actions;

	\item $\mathcal{Z}=\{1, \ldots,|\mathcal{Z}|\}$ is the set of observations for the partially observable state $e\in \mathcal{E}$;
			
	\item The function $T_s:\mathcal{S}\times \mathcal{E} \times \mathcal{A}\times \mathcal{S}\rightarrow [0,1]$ describes the probability of transitioning to a state $s'$ given the action $a$ and the system's state $(s,e)$, i.e., $T_s(s, e, a, s')\!:=\!\mathbb{P}(s_{k+1}\!=\!s'|s_{k}\!=\!s, e_{k}\!=\!e,a_{k}\!=\!a)$;

    \item The function $T_e: \mathcal{S} \times\mathcal{E}\times \mathcal{A}\times \mathcal{S}\times \mathcal{E}\rightarrow [0,1]$ describes the probability of transitioning to a state $e'$ given the action $a$, the successor observable state $s'$ and the system's current state $(s,e)$, i.e., $T_e(s, e, a, s', e')
	    :=\!\mathbb{P}(e_{k+1}\!=\!e'|s_{k}\!=\!s, e_{k}\!=\!e,a_{k}\!=\!a, s_{k+1}\!=\!s')$;

	\item The function $O:\mathcal{S}\times \mathcal{E} \times \mathcal{A} \times \mathcal{Z} \rightarrow [0,1]$ describes the probability of observing the measurement $z\in \mathcal{Z}$, given the current state of the system $(s',e') \in \mathcal{S} \times \mathcal{E}$ and the action $a$ applied at the previous time step, i.e., $O(s',e', a, z) :=  P(z_k\!=\!z|s_{k}\!=\!s', e_k\!=\! e',a_{k-1}\!=\!a)$;
\end{itemize}
%\end{definition}
MOMDPs were introduced in~\cite{ong2010planning} to model systems where a subspace of the state space is perfectly observable. The advantage of distinguishing between fully and partially observable states is that a belief state is needed only for the partially observable states. Thus, we introduce the belief vector $\be \in \Be=\{ \be \in \mathbb{R}^{|\mathcal{E}|} : \sum_{e=1}^{|\mathcal{E}|} \be(e) = 1\},$
where each entry $\be(e)$ represents the posterior probability that the partially observable state $e_k$ equals $e \in \mathcal{E}$.
% , i.e., $\be(e) = \mathbb{P}(e_k = e|a_0,...,a_{k-1},s_0,...,s_k,z_0,...,z_k),~\forall e \in \{1\ldots,|\mathcal{E}|\}$. 

\vspace{2pt}
\noindent
\textbf{Syntactically Co-Safe LTL Specifications}\\
We consider objectives which are expressed using scLTL specifications. An scLTL specification is defined as follows:
\begin{equation*}
    \psi : =  p ~ | ~ \neg p ~ | ~ \psi_1 \land \psi_2 ~ | ~ \psi_1 \lor \psi_2 ~ | ~ \psi_1 U \psi_2~ | ~ \bigcirc \psi ,
\end{equation*}
where the atomic proposition $p \in \{\texttt{true}, \texttt{false}\}$ and $\psi, \psi_1, \psi_2$ are scLTL formulas, which can be defined using the logic operators negation ($\neg$), conjunction ($\land$) and disjunction ($\lor$). Furthermore, scLTL formulas can be specified using the temporal operators until ($U$) and next ($\bigcirc$). Each atomic proposition $p_i$ is associated with a subset of the MOMDP state space $\mathcal{P}_i \subset \mathcal{S} \times \mathcal{E}$ and, for the MOMDP state $\omega_k=(s_k, e_k)$, the proposition $p_i$ is $\texttt{true}$ if $\omega_k \in \mathcal{P}_i$. Finally, satisfaction of a specification $\psi$ for the trajectory $\boldsymbol{\omega}_k = [\omega_k, \omega_{k+1}, \ldots]$, denoted by  
\begin{equation*}\label{eq:specFormula}
    \boldsymbol{\omega}_k \models \psi
\end{equation*}
is recursively defined as follows: $i)$ $\boldsymbol{\omega}_k \models p \iff$ $\omega_k \in \mathcal{P}$, $ii)$ $\boldsymbol{\omega}_k \models \psi_1 \land \psi_2 \iff$ $(\boldsymbol{\omega}_k \models \psi_1)\land (\boldsymbol{\omega}_k \models \psi_1)$, $iii)$ $\boldsymbol{\omega}_k \models \psi_1 \lor \psi_2 \iff$ $(\boldsymbol{\omega}_k \models \psi_1)\lor (\boldsymbol{\omega}_k \models \psi_1)$, $iv)$ $\boldsymbol{\omega}_k \models \psi_1 U \psi_2 \iff$    $\boldsymbol{\omega}_l \models \psi_2$ and $\boldsymbol{\omega}_j \models \psi_2,~\forall j \in \{k, \ldots, l-1\}$, $v)$ $\boldsymbol{\omega}_k \models \bigcirc \psi \iff$ $\boldsymbol{\omega}_{k+1} \models \psi$.

% \begin{assumption}
% Given a set $\mathcal{P}_i$ associated with an atomic proposition $p_i$, a set of observations $\{z_0, \ldots, z_k\}$ and a set of fully observable states $\{s_0, \ldots, s_k\}$, we can determine if the MOMDP state $\omega_k = (s_k, e_k) \in \mathcal{P}_i$.
% \end{assumption}

% Notice that the above assumption is satisfied in several robotic applications. For instance, when an agent has to avoid partially observable obstacles, it is reasonable to assume that collisions can be detected from measurements.

\begin{assumption}
We consider reachability specifications, which are satisfied when the observable state $s \in \mathcal{S}$ of a MOMDP $\mathcal{M}$ reaches a target set $\mathcal{T} \subset \mathcal{S}$.
\end{assumption}

The above assumption is not restrictive, as the problem of checking if a trajectory of a MOMDP satisfies any scLTL specification can be recasted as a reachability problem over an extended MOMDP. Please refer to~\cite[Chapter~3]{belta2017formal},\cite{ bouton2020point, nilsson2018toward} for further details on how to construct such extended~MOMDP.

\section{Time-Optimal Quantitative MOMDP}
% \vspace{2pt}
\noindent
\textbf{Problem Formulation}\\
In this section, we introduce the problem under study. Given a MOMDP $\mathcal{M}$ with observable states $\mathcal{S}$, partially observable states $\mathcal{E}$, and target set $\mathcal{T}$ associated with the specification $\psi$, we consider the finite-horizon problem of maximizing the probability of satisfying the specification $\psi$, while minimizing the expected time to complete the task. In particular, we define the following time-optimal quantitative constrained MOMDP (CMOMDP) 
\begin{equation}\label{eq:minTimeQua}
\begin{aligned}
\pi^{\textrm{TOQ}}=\underset{\pi}{\text{argmin}} ~~ & \mathbb{E}^\pi \Bigg[ \sum_{t = 0}^N \mathds{1}_{\mathcal{S}\setminus\T}(s_t) \Bigg]   \\
 \text{subject to} ~~ ~& \pi \in \argmax_{\kappa} \mathbb{P}^\kappa[\boldsymbol{\omega}  \models \psi], %, \forall t \in \{0,\ldots, N\!-\!1\}
\end{aligned}
\end{equation}
where $\mathbb{E}^\pi[\cdot]$ denotes the expectation under the policy $\pi$, $N$ represents the duration of the task and the indicator function $\mathds{1}_{\mathcal{S}\setminus\T}(s) = 1$ when $s \in \mathcal{S}\setminus\T$ and  $\mathds{1}_{\mathcal{S}\setminus\T}(s) = 0$ when $s \notin \mathcal{S}\setminus\T$. 
In the above problem, $\mathbb{P}^\kappa[\boldsymbol{\omega} \models \psi]$ represents the probability that the closed-loop trajectory under the policy $\kappa: \mathcal{S}  \times \Be \rightarrow \mathcal{A}$ will satisfy the specifications. Therefore, the optimal policy $\pi^{\textrm{TOQ}}: \mathcal{S}  \times \Be \rightarrow \mathcal{A}$ from~\eqref{eq:minTimeQua} maximizes the probability of satisfying the specifications while minimizing the expected time to complete the control task, i.e., reaching the set $\T\times\mathcal{E} \subset \mathcal{S} \times \mathcal{E}$.

\vspace{2pt}
\noindent
\textbf{Motivating Example}\\
Problem~\eqref{eq:minTimeQua} is motivated by the example shown in Figure~\ref{fig:closedLoopMotivatinComparison}, where a Segway has to collect science samples which may be located in the goal region (green) while avoiding known obstacle regions (dark brown) and exploring uncertain regions (light brown). 
The control problem can be formulated as a MOMDP, where the Segway's position is perfectly observed and only partial observations about the traversability of the uncertain regions are available. Figure~\ref{fig:closedLoopMotivatinComparison} shows an example with one goal region $\mathcal{G}$ and four uncertain regions $\mathcal{R}_1$, $\mathcal{R}_2$, $\mathcal{R}_3$, and $\mathcal{R}_4$, which may be traversable with probability $0.9$, $0.4$, $0.3$, and $0.5$, respectively. The Segway receives a perfect measurement when it is next to an uncertain region, otherwise the measurement is corrupted as described in the example section. In this example, the task has a duration of $N=30$ time steps. The control objective is given by the scLTL formula $\psi = \neg \texttt{Collision} U \texttt{Goal}$, where the atomic proposition $\texttt{Collision}$ is $\texttt{true}$ when the system is in a cell occupied by an obstacle and the atomic proposition $\texttt{Goal}$ is $\texttt{true}$ when the system reached the goal. 
\begin{figure}[t!]
    \centering
	\includegraphics[width= 0.86\columnwidth]{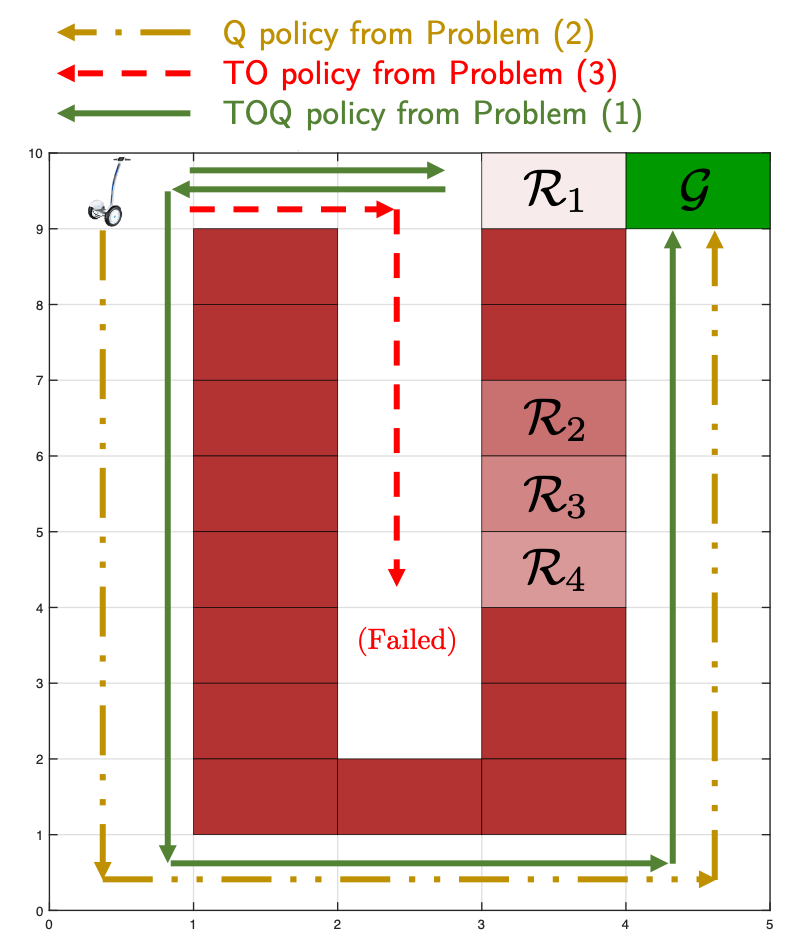}
    \caption{This figure shows a navigation example with several obstacles (dark brown), one goal region $\mathcal{G}$ (green) and four uncertain regions (light brown) $\mathcal{R}_1$, $\mathcal{R}_2$ $\mathcal{R}_3$, and $\mathcal{R}_4$. In this example, all uncertain regions are not traversable.}\label{fig:closedLoopMotivatinComparison}
\end{figure}

As discussed in~\cite{vasile2016control,nilsson2018toward,haesaert2019temporal}, a control policy can be computed maximizing the probability that $\psi$ is satisfied, i.e.,
\begin{equation}\label{eq:maxSpecs}
\begin{aligned}
\pi^{\textrm{Q}} = \underset{\kappa}{\text{argmax}} \quad \mathbb{P}^\kappa[\boldsymbol{\omega} \models \psi]. %, \forall t \in \{0,\ldots, N\!-\!1\}
\end{aligned}
\end{equation}
% The solution to the above quantitative problem can be approximated
% The optimal policy from the above quantitative problem 
% can be approximated as shown in~\cite{haesaert2018temporal}  \cite{nilsson2018toward} \cite{haesaert2019temporal}. 
% can be approximated using finite state controllers~\cite{ahmadi2020stochastic} and point-based methods~\cite{haesaert2019temporal}  \cite{vasile2016control}.  
Alternatively, a control policy can be synthesized minimizing the expected time to complete the task. The time-optimal problem is given by the standard reward minimization:
\begin{equation}\label{eq:minTime}
\begin{aligned}
\pi^{\textrm{TO}} = \underset{\pi}{\text{argmin}} ~~ & \mathbb{E}^\pi \Bigg[\sum_{t = 0}^N \mathds{1}_{\mathcal{S}\setminus\T}(s_t) \Bigg]   \\
\end{aligned},
\end{equation}
where the indicator function $\mathds{1}_{\mathcal{S}\setminus\T}$ is defined as in~\eqref{eq:minTimeQua}.
Notice that the solution to the above minimization problem can be approximated with point-based methods~\cite{pineau2003point} or finite state controllers~\cite{PoupartB03}.

Figure~\ref{fig:closedLoopMotivatinComparison} shows the closed-loop behaviors associated with the control policies from Problems~\eqref{eq:minTimeQua}--\eqref{eq:minTime}.
The Time-Optimal (TO) policy from Problem~\eqref{eq:minTime} steers the system beside the uncertain regions to collect perfect measurements about the traversability of the terrain. In this example, all uncertain regions are not traversable and therefore the control policy from Problem~\eqref{eq:minTime} fails to reach the goal state $\mathcal{G}$ in $N=30$ time steps. 
On the other hand, the Time-Optimal Quantitative (TOQ) policy from Problem~\eqref{eq:minTimeQua} and the Quantitative (Q) policy from Problem~\eqref{eq:maxSpecs}, which are designed to maximize the probability of satisfying the specification, reach the goal set
$\G$. Finally, we notice that the TOQ policy first explores region~$\mathcal{R}_1$ and then takes the path around the obstacle to reach the goal. This behavior minimizes the expected time to complete the task, as region~$\mathcal{R}_1$ may be traversable with probability $0.9$.
% the expected time to complete the task associated with the TOQ policy is lower than the one associated with the Q policy, because the TOQ policy from Problem~\eqref{eq:minTimeQua} explores region~$\mathcal{R}_4$ before reaching the goal. 
Thus, this example shows the advantage of synthesizing TOQ policies, which minimize the expected time to complete the task, while guaranteeing that the probability of satisfying the specifications is maximized.

\section{Exact Dynamic Programming Update}
In this section, we first show that the optimal value function $V^*_k(s, \cdot):\Be \rightarrow \mathbb{R}$ of the time-optimal quantitative problem~\eqref{eq:minTimeQua} is piecewise affine for all $s \in \mathcal{S}$ and $k \in \{0,\ldots,N\}$. Afterwards, following the approach presented in~\cite{isom2008piecewise}, we define a pair of support vectors which characterize the optimal value function $V^*_k(s, \cdot)$ for all $s \in \mathcal{S}$ and $k \in \{0,\ldots,N\}$.

As shown in~\cite{ong2010planning}, the  synthesis problem can be reformulated as a stochastic optimal control problem for a fully observable uncertain system, where the states are the belief $\be$ and the fully observable state $s$ of the MOMDP.
% \footnote{In the case of partially observable Markov decision processes (POMDP) the problem is reformulated as a stochastic control problem over the belief state $b$. In the MOMDP settings we only need to maintain a belief for the partially observable state $e$, as the state $s$ is perfectly observed~\cite{ong2010planning}.}.
Indeed, the belief evolves accordingly to the following update equation: 
\begin{equation}\label{eq:beliefUpdateSingle}
\begin{aligned}
    \be'(e') =&\eta O(s',e', a, z) \\
    &\times \sum_{e\in\mathcal{E}}T_s(s,e,a,s')T_e(s,e,a,s',e')\be(e),
\end{aligned}
\end{equation}
where the scalar $\eta = 1/ P(z, s' | s, \be, a)$ is a normalization constant~\cite{ong2010planning, peron2017fast}.

Next, we introduce two lemmas that allow us to reformulate Problem~\eqref{eq:maxSpecs} and Problem~\eqref{eq:minTime} as standard reward maximization problems. Afterwards, we will leverage these results to derive an exact dynamic programming update for the time-optimal quantitative Problem~\eqref{eq:minTimeQua}.

\begin{lemma}\label{lemma:qual}
Consider a MOMDP $\mathcal{M}$ with terminal set $\T\times\mathcal{E}$ and a finite horizon $N$. The probability that the quantitative policy $\pi^{\textrm{Q}}$ from Problem~\eqref{eq:minTime} satisfies the specifications~is
\begin{equation*}
\underset{\kappa}{\text{max}}~ \mathbb{P}^\kappa[\boldsymbol{\omega} \models \psi] = \bar J_0(s, \be),
\end{equation*}
where the optimal value function $\bar J_0$ is given by the following dynamic programming recursion:
\begin{equation}\label{eq:optConstrVal}
    \begin{aligned}
    \bar{J}_k(s,\be) &
    = \mathds{1}_{\T}(s) + \mathds{1}_{\mathcal{S}\setminus\T}(s) \max_{a\in\mathcal{A}}\mathbb{E}[\bar J_{k+1}(s',\be')|s,\be,a]\\
    % &= \max \Big[ \mathds{1}_{\T}(s), \max_{a \in \mathcal{A}} \mathbb{E}[\bar J_{k+1}(s', \be')|s,\be] \Big].
    \end{aligned}
\end{equation}
with $\bar J_N(s, \cdot) = \mathds{1}_{\T}(s)$ for all s $\in \mathcal{S}$. Furthermore, the optimal value function $\bar{J}_k(s,\cdot):\Be \rightarrow [0,1]$ is piecewise-affine for all $k\in\{0,\ldots,N\}$ and for all $s\in\mathcal{S}$.
\end{lemma}
\vspace{0.2cm}
\begin{proof}
Notice that, given a policy $\kappa : \mathcal{S} \times \Be \rightarrow \mathcal{A}$, the probability of satisfying the specification is given by the probability of reaching the terminal set $\T$, i.e.,
\begin{equation*}\label{eq:reachProbSt}
\begin{aligned}
    \mathbb{P}^\kappa[\boldsymbol{\omega} \models \psi] &=\mathbb{P}^\kappa\big[ \exists k\in\{0,\ldots, N\}: s_k \in \T,  \\
    & \quad \quad \quad s_j \in \mathcal{S}\setminus\T, \forall j \in\{0,\ldots, k-1\} \big] \\
    &=\mathbb{E}^\kappa\Bigg[ \sum_{j=0}^N \Bigg( \prod_{i=0}^{j-1} \mathds{1}_{\mathcal{S} \setminus \T}(s_i) \Bigg) \mathds{1}_{\T}(s_j) \Bigg],
\end{aligned}
\end{equation*}
where $\mathbb{E}^\kappa[\cdot]$ denotes the expectation under the policy $\kappa$. For more details on the above stochastic reachability problem please refer to~\cite{summers2010verification}.
% Furthermore, from~\cite[Lemma~4]{summers2010verification}, we have that the value function associated with the above probability can be computed by the following recursion:
% \begin{equation*}
% \begin{aligned}
%     J^\kappa_k(s,\be) & = \mathds{1}_{\T}(s) + \mathds{1}_{\mathcal{S}\setminus\T}(s) \mathbb{E}^\kappa[J_{k+1}^\kappa(s',\be')|s,\be]
% \end{aligned}
% \end{equation*}
% with $J^\kappa_N(s,\cdot)=1$ if $s\in \T$ and $J^\kappa_N(s,\cdot)=0$ if $s\notin \T$.
% \begin{equation*}
%     \begin{aligned}
%             J^\kappa_N(s,\be) = \begin{cases} 1 = \sum_{i = 1}^{|\mathcal{E}|}\be(i) & \mbox{ If } s\in \T \\
%     0 & \mbox{ Else } 
%         \end{cases}
%     \end{aligned}.
% \end{equation*}

% By standard dynamic programming arguments and the above equation, we have that the optimal value function for the quantitative Problem~\eqref{eq:maxSpecs} is given by
Furthermore from~\cite[Theorem~4]{summers2010verification}, we have that the optimal value function $\bar J_k: \mathcal{S} \times \Be \rightarrow \mathbb{R}$, which is associated with the optimal policy that maximizes the probability of reaching the set $\T$, is given by the following recursion
\begin{equation}\label{eq:valStep1}
\begin{aligned}
    \bar J_k(s,\be) &= \mathds{1}_{\T}(s) + \mathds{1}_{\mathcal{S}\setminus\T}(s) \max_{a\in\mathcal{A}}\mathbb{E}[\bar J_{k+1}(s',\be')|s,\be,a]
\end{aligned}
\end{equation}
where $\bar J_N(s,\cdot)= \mathds{1}_\T(s)$. 
Next, we show by induction that $\bar J_k(s,\cdot):\Be \rightarrow \mathbb{R}$ is piecewise affine for all $k \in \{0,\ldots, N\}$ and $s\in\mathcal{S}$. Assume that $\bar J_{k+1}(s, \cdot)$ is piecewise affine for a set of support vectors $\Lambda_s$, i.e., $\bar J_{k+1}(s, \be) = \max_{\beta \in \Lambda_s} \beta^\top \be =\max_{\beta \in \Lambda_s}\sum_e \beta(e) \be(e) $. 
Then, using the belief update~\eqref{eq:beliefUpdateSingle} and the definition of the value function $\bar J_{k+1}(s, \cdot)$, we have that
\begin{equation}\label{eq:midStep}
\begin{aligned}
    &\mathbb{E}[\bar J_{k+1}(s',\be')|s,\be,a] = \sum_{s',z} P(s',z|s,\be,a) J_{k+1}(s',\be')\\
    &= \sum_{s',z} P(s',z|s,\be,a) \max_{\beta \in \Lambda_s}  \sum_{e'}  \beta(e') \be'(e') \\
    &= \sum_{s',z} \frac{1}{\eta} \max_{\beta \in \Lambda_s}  \sum_{e'} \beta(e') \eta O(s',e',a,z) \sum_{e} T_s(s,e,a,s')\\
    &\qquad \qquad \qquad \qquad \qquad \qquad \qquad \times T_e(s,e,a,s',e')\be(e).
\end{aligned}
\end{equation}
Now define $\beta'(e) = \sum_{e'} F(s,e,a,s',e',z)\beta(e')$ for
\begin{equation}\label{eq:Fscalar}
\begin{aligned}
F(s,e,a,s',e',z)= T_s(s,e,a,s')&T_e(s,e,a,s',e')\\
&~~\times O(s',e',a,z),
\end{aligned}
\end{equation} 
then equation~\eqref{eq:midStep} can be rewritten as 
\begin{equation}\label{eq:DPupdate}
\begin{aligned}
    \mathbb{E}[&\bar J_{k+1}(s',\be')|s,\be,a] \\
    & = \sum_{s',z} \max_{\beta \in \Lambda_s} \sum_{e} \be(e) \sum_{e'} F(s,e,a,s',e',z)\beta(e') \\
    &= \sum_{s',z} \max_{\beta \in \Lambda_s} \be^\top \beta'.
\end{aligned}
\end{equation}
Equation~\eqref{eq:DPupdate} implies that the conditional expectation in~\eqref{eq:valStep1} is a piecewise affine function of $\be$. Therefore, $\bar J_{k}(s,\cdot)$ is piecewise affine as it is given by the summation and the point-wise maximization of picecewise affine functions for all $ s \in \mathcal{S}$. The proof is concluded by induction on $k$ as the value function $\bar J_N(s,\cdot)$ is piecewise affine for all $s \in \mathcal{S}$.
\end{proof}

\begin{lemma}\label{lemma:maxProblem}
Consider a MOMDP $\mathcal{M}$ with terminal set $\T\times\mathcal{E}$ and a finite horizon $N$. The optimal control policy $\pi^{\textrm{TO}}$ from Problem~\eqref{eq:minTime} is the optimizer of the following problem:
\begin{equation*}\label{eq:maxGoal}
    \max_\pi~ \mathbb{E}^\pi \bigg[\sum_{t = 0}^N \mathds{1}_{\T}(s_t) \bigg].
\end{equation*}
% Furthermore, we have that
% \begin{equation*}
%     \min_\pi~ \mathbb{E}^\pi \Bigg[\sum_{t = 0}^N \mathds{1}_{\mathcal{S}\setminus\T}(s_t) \Bigg]+\max_\pi ~\mathbb{E}^\pi \Bigg[\sum_{t = 0}^N \mathds{1}_{\T}(s_t) \Bigg] = N + 1.
% \end{equation*}
\end{lemma}
% \vspace{0.2cm}
\begin{proof} Notice that by definition $\mathds{1}_{\mathcal{S} \setminus \T}(s_t) = 1-\mathds{1}_{\T}(s_t), \forall s \in \mathcal{S}$. Thus, we have that 
\begin{equation*}
\begin{aligned}
\argmin_\pi \mathbb{E}^\pi \bigg[\sum_{t = 0}^N \mathds{1}_{\mathcal{S} \setminus \T}(s_t) \bigg] \! = \! \argmin_\pi \mathbb{E}^\pi \bigg[\sum_{t = 0}^N ( 1-\mathds{1}_{\T}(s_t) ) \bigg]\\
\! = \! \argmin_\pi \mathbb{E}^\pi \bigg[\sum_{t = 0}^N - \mathds{1}_{\T}(s_t) \bigg] = 
\argmax_\pi \mathbb{E}^\pi \bigg[\sum_{t = 0}^N \mathds{1}_{\T}(s_t) \bigg],
\end{aligned}
\end{equation*}
which concludes the proof.
\end{proof}
% \vspace{0.2cm}

\vspace{2pt}
\noindent
\textbf{Optimal Value Function}\\
% As we have discussed, Problem~\eqref{eq:minTimeQua} is a constrained reachability problem over the MOMDP. In Lemma~\ref{lemma:qual}, we have shown that the optimal value function for the quantitative Problem~\eqref{eq:maxSpecs} is piecewise-affine and we know that also the time-optimal value function for Problem~\eqref{eq:maxGoal} is piecewise-affine~\cite{sondik1978optimal}. Therefore, we can modify the strategy presented in~\cite{isom2008piecewise} to solve the constrained MOMDP from~\eqref{eq:minTimeQua}. 
In what follows, we leverage the dynamic programming update from Lemma~\ref{lemma:qual} and the maximization problem from Lemma~\ref{lemma:maxProblem} to design an exact dynamic programming update for the time-optimal quantitative problem~\eqref{eq:minTimeQua}.
We modify the strategy presented in~\cite{isom2008piecewise} to solve the CMOMDP from~\eqref{eq:minTimeQua}. The key idea is to construct a set of vector pairs $\langle \alpha_{s,k}^{i}, \beta_{s,k}^{i} \rangle$, which define the optimal value function 
% $V_k^*: \mathcal{S}\times \Be \rightarrow \mathbb{R}$ 
% . Indeed, from Lemma~\ref{lemma:maxProblem} we have that the value function associated with the quantitative problem is piecewise affine and we know from~\cite{isom2008piecewise} that the value function , thus the optimal value function $V_k^*: \mathcal{S}\times \Be \rightarrow \mathbb{R}$ can be written as
\begin{equation}\label{eq:valFuns}
\begin{aligned}
V_k^*(s,\be)  = \max_{\langle \alpha, \beta \rangle \in \Gamma_{s,k}^*}\quad & \alpha^\top \be \\    \text{subject to} \quad~ &\langle \alpha, \beta \rangle \in \argmax_{\langle \alpha, \beta \rangle \in \Gamma_{s,k}^*} \beta^\top \be, \\    
\end{aligned}
\end{equation}
where at time $k$ the set $\Gamma_{s,k}^*$ collects the support vector pairs associated with the observable state $s\in\mathcal{S}$. 

The support vectors can be updated using the following recursion:
\begin{equation}\label{eq:alphaBetaUpdate}
    \begin{aligned}
    \alpha_{s, k}^{a,z, s',i}(e) &=  \frac{\mathds{1}_{\G}(s)}{|\mathcal{Z}||\mathcal{S}|} +  \sum_{e'}  F(s, e, a, s', e', z)\alpha_{s', k+1}^{i}(e'),\\
        \beta_{s, k}^{a,z,s',i}(e) &= 
    \frac{\mathds{1}_{\G}(s)}{|\mathcal{Z}||\mathcal{S}|} +
    \mathds{1}_{\mathcal{Q}\setminus\G}(s) \sum\limits_{e'}  F(s, e, a, s', e', z)\\
    &\qquad\qquad\qquad\qquad\qquad\qquad\qquad \times\beta_{s', k+1}^{i}(e')\\
    \Gamma^{*,a}_{s,k} &= \oplus_{z \in \mathcal{Z}, s'\in\mathcal{S}} \{\langle \alpha_{s, k}^{a,z,s',i}, \beta_{s, k}^{a,z,s',i} \rangle | \\
    &\qquad\qquad\qquad\qquad\qquad\forall i \in \{1,\ldots, |\Gamma_{s,k+1}^*| \}\} \\
    \Gamma_{s,k}^* &= \cup_{a\in\mathcal{A}} \Gamma^{*,a}_{s,k}.
    \end{aligned}
\end{equation}
where the function $F$ is defined as in \eqref{eq:Fscalar} and 
\begin{equation}\label{eq:gammaInit}
    \begin{aligned}
        \Gamma_{s,N}^*=\begin{cases} \langle 1_{|\mathcal{Z}|}, 1_{|\mathcal{Z}|} \rangle & \mbox{ if } s \in \T, \\
        \langle 0_{|\mathcal{Z}|}, 0_{|\mathcal{Z}|} \rangle & \mbox{ otherwise}.
        \end{cases}
    \end{aligned}
\end{equation}
The backup update of the $\alpha$-vector is used to compute the support vectors associated with the cost and it was presented in~\cite{araya2010closer}.
On the other hand, the backup update of the $\beta$-vector
is designed based on the dynamic programming update~\eqref{eq:optConstrVal} from Lemma~\ref{lemma:qual} and it is a key contribution of this work.
The following lemma illustrates that the backup update of the $\beta$-vector, which defines the set of support vectors $\Gamma_{s,k}^*$ from~\eqref{eq:alphaBetaUpdate}, allows us to compute the probability that the time-optimal quantitative policy satisfies the specifications. 
% Therefore, the constraint from the optimal value function~\eqref{eq:valFuns} encodes the quantitative constraint from Problem~\eqref{eq:minTimeQua}. 

\begin{lemma}\label{th:constrValFun}
Let $\Gamma_{s,k}^{*}$ be the set of support vectors constructed using the dynamic programming recursion from~\eqref{eq:alphaBetaUpdate}. Then, the constraint value function
\begin{equation}\label{eq:optConstr}
    J^*_{k}(s, \be) =  \max_{\langle \alpha, \beta \rangle \in\Gamma_{s,k}^{*}} \beta^\top \be
\end{equation}
represents the probability that the time-optimal quantitative policy $\pi^\textrm{TOQ}$ satisfies the specifications, i.e., $J^*_{k}(s, \be) = \mathbb{P}^{\pi^\textrm{TOQ}}\big[[\omega_k,\ldots, \omega_N] \models \psi \big], \forall k\in \{0,\ldots, N\}$.
\end{lemma}
% \vspace{0.2cm}
\begin{proof}
% As in the POMDPs settings~\cite[Theorem~7.4.1]{krishnamurthy2016partially}, we will represent the value function as the solution of a linear program. 
% By definition we have that the the time-optimal quantitative policy~\eqref{eq:minTime} maximizes the probability of satisfying the specifications. Therefore, from Lemma~\ref{lemma:maxProblem} we have that 
% \begin{equation*}
%     \bar J_{k}(s, \be) = \mathbb{P}^{\pi^\textrm{TOQ}}\big[[\omega_k,\ldots, \omega_N] \models \psi \big], \forall k\in \{0,\ldots, N\}.
% \end{equation*}
First, we show by induction that $\bar J_{k}(s, \cdot) =  J_{k}^*(s, \cdot)$ for all $s\in\mathcal{S}$. Assume that $\Gamma_{s,k+1}^{*} = \Lambda_s$, which implies that
\begin{equation}\label{eq:indAss}
    J^*_{k+1}(s, \be)=\max_{\langle \alpha, \beta \rangle \in \Gamma_{s,k+1}^{*}} \be^\top \beta= \max_{ \beta \in \Lambda_s} \be^\top \beta = \bar J_{k+1}(s, \be).
\end{equation}
Then, from equations~\eqref{eq:valStep1} and~\eqref{eq:DPupdate}, we have that 
\begin{equation*}
\begin{aligned}
    \bar J_k(&s, \be) = \mathds{1}_\T(s) + \mathds{1}_{\mathcal{S} \setminus \T}(s) \max_{a \in \mathcal{A}} \sum_{s',z} \max_{\beta \in \Lambda_s} \be^\top \beta'\\
    & = \max_{a\in\mathcal{A}}\sum_{s',z} \Bigg[ \frac{\mathds{1}_\T(s)}{|\mathcal{Z}||\mathcal{S}|} + \mathds{1}_{\mathcal{S} \setminus \T}(s) \max_{\beta \in \Lambda_s}  \be^\top \beta'\Bigg]\\
     & = \max_{a\in\mathcal{A}}\sum_{s',z} \max_{\beta \in \Gamma_{s,k+1}^{*}}\Bigg[  \frac{\mathds{1}_\T(s)}{|\mathcal{Z}||\mathcal{S}|}1_{|\mathcal{E}|}^\top + \mathds{1}_{\mathcal{S} \setminus \T}(s)  (\beta')^\top \Bigg] \be \\
     &= \max_{a\in\mathcal{A}} \max_{\langle \alpha, \beta \rangle \in \Gamma_{s,k}^{*,a}} \beta^\top \be =\max_{\langle \alpha, \beta \rangle \in \Gamma_{s,k}^{*}} \beta^\top \be = J_k^*(s, \be),
\end{aligned}
\end{equation*}
where  $\beta'(e) = \sum_{e'} F(s,e,a,s',e',z)\beta(e')$, $1_{|\mathcal{E}|} \in \mathbb{R}^{|\mathcal{E}|}$ is a vector of ones, $1_{|\mathcal{E}|}^\top \be = 1$ and the sets of support vectors $\Gamma_{s,k}^{*,a}$ and $\Gamma_{s,k}^{*}$ are defined by the backup update~\eqref{eq:alphaBetaUpdate} for the set of support vectors $\Gamma_{s,k+1}^*$ from equation~\eqref{eq:indAss}.
% . In the above derivation, we leveraged the assumption on the set of support vectors $\Gamma$ which define both $J^*_{k+1}(s, \be)$ and $\bar J_{k+1}(s, \be)$. 
Finally, as $J^*_{N}(s, \cdot)=\bar J_{N}(s, \cdot), \forall s \in \mathcal{S}$ by induction we have that $J^*_{k}(s, \cdot)=\bar J_{k}(s, \cdot), \forall s \in \mathcal{S}~ \forall k \in \{0, \ldots, N\}$, which together with Lemma~\ref{lemma:maxProblem} and the definitions of Problems~\eqref{eq:minTimeQua} and~\eqref{eq:maxSpecs} imply that 
$J^*_{k}(s, \cdot)=\bar J_{k}(s, \cdot) = \mathbb{P}^{\pi^\textrm{Q}}\big[[\omega_k,\ldots, \omega_N] \models \psi \big]=\mathbb{P}^{\pi^\textrm{TOQ}}\big[[\omega_k,\ldots, \omega_N] \models \psi \big]$, $\forall k \in \{0, \ldots, N\}$ and $\forall s \in \mathcal{S}$.
\end{proof}
\vspace{0.2cm}

% Finally, we notice that at each time step the dynamic programming update from~\eqref{eq:alphaBetaUpdate} generates in the worst case $|\mathcal{A}| |\Gamma_{s,k+1}^{*}|^{|\mathcal{Z}|}$ new support vector pairs. Therefore, the computational complexity grows exponentially at each update. In the next section, we propose a point-based strategy, which allows us to approximate the optimal solution using a constant number of support vectors. 

\section{Point-Based Approximation}\label{sec:pointBased}
At each time step the dynamic programming update from~\eqref{eq:alphaBetaUpdate} generates in the worst case $|\mathcal{A}| |\Gamma_{s,k+1}^{*}|^{(|\mathcal{Z}|+|\mathcal{S}|)}$ new support vector pairs~\cite{pineau2003point}. 
% Therefore, the computational complexity grows exponentially.
In this section, we present a point-based update, where the optimal value function is approximated by a constant number of vectors computed for a set $\mathcal{D}_b=\{\be^{(1)}, \ldots, \be^{(n)}\}$ of $n$ discrete beliefs. 

The proposed point-based strategy is based on the update from equation~\eqref{eq:alphaBetaUpdate}. In particular, Algorithm~\ref{algo:backup} computes one pair of vectors $\langle \alpha^{a^*}, \beta^{a^*} \rangle$ that approximates the optimal value function~\eqref{eq:valFuns} at a point $(s,\be)$. In line~1 of Algorithm~\ref{algo:backup}, we compute the active pair of support vectors using the following expression\footnote{
% Compared to the backup update from~\cite{ong2010planning}, we propagate the belief through the cost and the constraint. 
For more details on the belief propagation please refer to~\cite{ong2010planning}.}
\begin{equation}\label{eq:updateVectors}
\begin{aligned}
\langle \alpha^{s', a, z}, \beta^{s', a, z}\rangle ~~ & \\
= \!\!\! \argmax_{\langle \alpha, \beta \rangle \in \Gamma_{s,k+1}}~ & \alpha^\top F_v(s,\be, a, s', z) \\    \text{subject to} \quad~ &\langle \alpha, \beta \rangle \in \!\!\! \argmax_{\langle \alpha, \beta \rangle \in \Gamma_{s,k+1}} \beta^\top F_v(s,\be, a, s', z), \\    
\end{aligned}
\end{equation}
where $F_v:\mathcal{S}\times\Be\times\mathcal{A}\times\mathcal{S}\times\mathcal{Z}\rightarrow \Be$ is the belief vector update, i.e., $\be' = F_v(s,\be, a, s', z)$. Afterwards, we update the support vectors pair associated with an action~$a$ (line~2). These vectors are then used to compute the set of admissible actions~(line~3) and the optimal action $a^*$~(line~4). Finally, we add the optimal pair $\langle \alpha^{a^*}, \beta^{a^*} \rangle$ to the set of support vector pairs $\Gamma_{s,k}$, which approximate the optimal value function $V^*_k(s, \cdot)$ from~\eqref{eq:valFuns} for all $s \in \mathcal{S}$.

\begin{algorithm}[h!]
\SetAlgoLined
\KwIn{$s, \be, \Gamma_{s,k+1}, \Gamma_{s,k}$ }
For all  $s' \in \mathcal{S}$, $a \in \mathcal{A}$, $z\in\mathcal{Z}$ 
\mbox{$\langle \alpha^{s', a, z}, \beta^{s', a, z} \rangle \leftarrow$ from Equation~\eqref{eq:updateVectors} }\;
For all $a \in \mathcal{A}$, $e \in \mathcal{E}$
\mbox{$\alpha^{a}(e) \leftarrow \mathds{1}_{\G}(s) + \sum_{s',z,e'}F(s, e, a, s', e', z) \alpha^{s', a, z}(e)$} 
\mbox{$\beta^{a}(e) \leftarrow \mathds{1}_{\G}(s)+ \mathds{1}_{\mathcal{S}\setminus\G}(s)\sum_{s',z,e'}F(s, e, a, s', e', z) $} \phantom{spacespacespacespacespacespacespac} $\times\alpha^{s', a, z}(e)$ \;
Compute $\mathcal{C} = \argmax_{a \in \mathcal{A}}(\be^\top \beta^a)$ \;
Compute $a^* = \argmax_{a \in \mathcal{C}}(\be^\top \alpha^a)$ \;
Add $\langle \alpha^{a^*}, \beta^{a^*} \rangle$ to $\Gamma_{s,k}$ \;
\KwOut{$\Gamma_{s,k}$}
\caption{\texttt{Backup}, $(\alpha,\beta)$-vectors computation}\label{algo:backup}
\end{algorithm}

\begin{algorithm}[h!]
\SetAlgoLined
\KwIn{$\Gamma_{s,k+1}$ }
\For {$s \in \mathcal{S}$}
{Initialize $\Gamma_{s,k}=\O$\;
\For {$\be \in \mathcal{D}_b$}{
$\Gamma_{s,k}$ $\leftarrow$ \texttt{Backup}($s$, $\be$, $\Gamma_{s,k+1}$, $\Gamma_{s,k}$) \;}
}
\KwOut{$\Gamma_{s,k}$} \caption{Value function update}\label{algo:valFunUpdate}
\end{algorithm}

The \texttt{Backup} function from Algorithm~\ref{algo:backup} is used to update the sets $\Gamma_{s,k}$, which define the value function approximation:
\begin{equation}\label{eq:appValFun}
\begin{aligned}
V_k(s,\be)  = \max_{\langle \alpha, \beta \rangle \in \Gamma_{s,k}}\quad & \alpha^\top \be \\    \text{subject to} \quad~ &\langle \alpha, \beta \rangle \in \argmax_{\langle \alpha, \beta \rangle \in \Gamma_{s,k}} \beta^\top \be. \\    
\end{aligned}
\end{equation}
For time $k\in\{0,\ldots,N-1\}$, the sets $\Gamma_{s,k}$ are recursively computed using Algorithm~\ref{algo:valFunUpdate}, which for all $s\in\mathcal{S}$ computes the support vector pairs at all belief points $\be \in \mathcal{D}_b$. 
The recursion is initialized setting $\Gamma_{s,N}$ equal to $\Gamma^*_{s,N}$.
% from equation~\eqref{eq:gammaInit}.
% Notice that $V_k$ from~\eqref{eq:appValFun} is defined by $|\mathcal{D}_b|$ support vector pairs. Therefore, as discussed in~\cite{pineau2003point}, the computational complexity of this approximation is polynomial time. 
% The properties of the proposed approximation and comparison with standard point-based strategy for constrained POMDPs are discussed in the Section~\ref{sec:diss}.

\begin{figure*}[h!]
    \centering
	\includegraphics[width= 1.85\columnwidth]{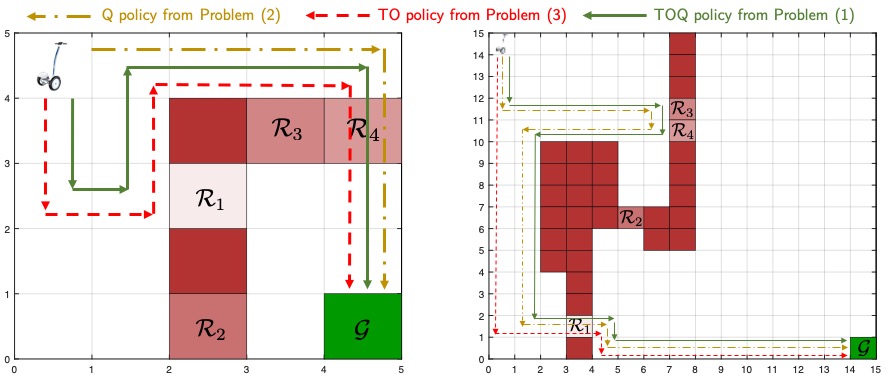}
    \caption{Grid worlds with several obstacles (dark brown), one goal region $\mathcal{G}$ (green) and four uncertain regions (light brown) $\mathcal{R}_1$, $\mathcal{R}_2$ $\mathcal{R}_3$, and $\mathcal{R}_4$, which may be free with probability $0.9$, $0.4$, $0.3$, and $0.5$. The closed-loop trajectories are associated with different environment realizations. In particular, $\mathcal{R}_4$ and $\mathcal{R}_1$ are traversable in the 5x5 and 15x15 grid worlds, respectively.}
    \label{fig:marsMission2}
\end{figure*}

Finally, we show that the sets of support vector pairs $\Gamma_{s,k}$ computed using Algorithms~\ref{algo:backup} and~\ref{algo:valFunUpdate} allow us to define an approximated constraint value function, which is a lower-bound of the probability of satisfying the specifications.
\vspace{0.2cm}
\begin{theorem}\label{th:appValFun}
Let $\Gamma_{s,k}$ be the set of support vectors constructed using the point-based strategy from Algorithms~\ref{algo:backup} and~\ref{algo:valFunUpdate}. Then, the approximated constraint value function
\begin{equation}\label{eq:appConstrValFun}
    J_{k}(s, \be) =  \max_{\langle \alpha, \beta \rangle \in\Gamma_{s,k}} \beta^\top \be
\end{equation}
is a lower-bound of the probability that the control policy $\pi^\textrm{TOQ}$ satisfies the specifications, i.e., $J_{k}(s, \be) \leq \mathbb{P}^{\pi^\textrm{TOQ}}\big[[\omega_k,\ldots, \omega_N] \models \psi \big], \forall k\in \{0,\ldots, N\}$.
\end{theorem}
\vspace{0.2cm}
\begin{proof}
% Notice that at each backup update in Algorithm~\ref{algo:backup} the vector $\beta^{a^*}$ maximizes the product $\be^\top \beta^{a^*}$. Therefore, from standard point-based arguments and the piecewise affine structure of the constraint value function $J^*_{k}(s,\be)$ from \eqref{eq:optConstr} we have 
The $\beta$-vectors computed by the backup Algorithms~\ref{algo:backup}-\ref{algo:valFunUpdate} are a subset of the $\beta$-vectors from~\eqref{eq:alphaBetaUpdate}, which define the optimal value function from~\eqref{eq:optConstr}. Therefore, as $\Gamma_{s,k} \subseteq \Gamma_{s,k}^*$ we have that
\begin{equation*}
   J_{k}(s,\be) = \max_{\langle \alpha, \beta \rangle \in\Gamma_{s,k}} \beta^\top \be \leq \max_{\langle \alpha, \beta \rangle \in\Gamma_{s,k}^*} \beta^\top \be = J^*_{k}(s,\be),
\end{equation*}
$\forall k \in \{0, \ldots, N\}, \forall s\in\mathcal{S}$ and $\forall \be \in \Be$. The above equation and Lemma~\ref{th:constrValFun} imply that $J_{k}(s,\be) \leq \mathbb{P}^{\pi^{\textrm{TOQ}}}\big[ [\omega_k,\ldots, \omega_N] \models \psi \big], \forall k \in \{0, \ldots, N\}, s\in\mathcal{S}$ and $\be \in \Be$.
\end{proof}

\section{Examples}
% We tested the proposed strategy on navigation problems, where the state of the robot is perfectly observable and the state of the environment is partially observable. First, we test the proposed strategy in grid worlds and we compare the performance and computational time with the time-optimal policy from~\eqref{eq:minTime} and the quantitative policy from~\eqref{eq:maxSpecs}. Afterwards, we leverage the proposed strategy on a continuous time example, where the time-optimal quantitative policy~\eqref{eq:minTimeQua} is used to compute high-level commands which are sent to a low-level time controller.

\subsection{Grid Worlds}
The proposed strategy is tested on three grid worlds shown in Figures~\ref{fig:closedLoopMotivatinComparison} and \ref{fig:marsMission2}. We compared the proposed Time-Optimal Quantitative (TOQ) policy approximated using a one-step look ahead and the value function from Section~\ref{sec:pointBased} with the Quantitative (Q) and Time-Optimal (TO) policies from Problems~\eqref{eq:maxSpecs}-\eqref{eq:minTime}, which are approximated using standard point-based approaches for reward maximization\footnote{Code available online: \href{https://github.com/urosolia/MOMDP}{\texttt{https://github.com/urosolia/MOMDP}}. All simulations are run on a 2015 Macbook Pro with a 2.5GHz Quad-Core Intel Core i7 and 16GB of memory.}. In all simulations, the Segway receives a perfect measurement when adjacent to an uncertain region, a measurement which is correct with probability $0.8$ when one grid cell away in the diagonal direction and an uninformative measurement otherwise. The uncertain regions $\mathcal{R}_1$, $\mathcal{R}_2$, $\mathcal{R}_3$, and $\mathcal{R}_4$, may be traversable with probability $0.9$, $0.3$, $0.4$, and $0.5$. Finally, in order to analyze the effect of the number of uncertain regions on the computational complexity, we also tested a scenario where region $\mathcal{R}_4$ is a known obstacle.

Figures~\ref{fig:closedLoopMotivatinComparison} and \ref{fig:marsMission2} show the closed-loop trajectories for different realizations of the uncertain regions. We notice that the TOQ policy behaves similar to the TO one, when the constraint from Problem~\eqref{eq:minTimeQua} does not restrict the search space. Consider the 5x5 grid world in Figure~\ref{fig:marsMission2}, where the agent can explore all uncertain regions in different orders, as the task horizon is $T=30$. In this example, the TOQ policy first explores region $\mathcal{R}_1$, and then it steers the agent through region $\mathcal{R}_4$. This behavior minimizes the expected time to complete the task, as region $\mathcal{R}_1$ has the highest probability of being free. Thus, the closed-loop trajectories associated with the TO and TOQ policies overlap. On the other hand, in the 15x15 grid world from Figure~\ref{fig:marsMission2}, the task horizon is $T=40$ and the agent cannot explore all regions. Therefore, the TOQ policy maximizes the number of visited regions and behaves as the Q policy. Indeed, in this 15x15 grid world, first visiting region $\mathcal{R}_1$, which has the highest probability of being free, would lead to a lower probability of mission success.
In general, the TOQ policy minimizes the expected time to complete the task, without compromising the probability of satisfying the specifications, as we have seen in Figure~\ref{fig:closedLoopMotivatinComparison}. \\

\begin{table}[t!]
\centering
\setlength\tabcolsep{2.5pt}
\begin{tabular}{lccccc} \midrule
\makecell{Grid \\ World}  & \makecell{Exp. \\ Time}  & \makecell{Prob. \\ Failure} & \makecell{Failure \\ Bound}  & \makecell{ Total \\  Time [s] } & \makecell{ Backup \\  Time [ms] }  \\ \midrule
$[$5x5$]_3^\mathrm{TO}$         & 8.12                          & 4.2\%     & N/A              & 4.09        & 0.45  \\[1.5pt]
$[$5x5$]_3^\mathrm{Q}$          & 27.78                         & 4.2\%    & $\leq$ 4.2\%      & 3.49       & 0.39  \\[1.5pt]
$[$5x5$]_3^\mathrm{TOQ}$        & 8.12                          & 4.2\%    & $\leq$ 4.2\%      & 8.29       & 0.92  \\[1.5pt]
\midrule
$[$5x5$]_4^\mathrm{TO}$         & 8.2                           & 2.1\%     & N/A               &7.09        & 0.62  \\[1.5pt]
$[$5x5$]_4^\mathrm{Q}$          & 28.39                         & 2.1\%     & $\leq$ 2.1\%      &6.57        & 0.58  \\[1.5pt]
$[$5x5$]_4^\mathrm{TOQ}$        & 8.2                           & 2.1\%     & $\leq$ 2.1\%      &16.86     & 1.48 \\[1.5pt]
\midrule
$[$10x5$]_3^\mathrm{TO}$        & 4.23                          & 4.2\%     & N/A               &4.69        & 0.33  \\[1.5pt]
$[$10x5$]_3^\mathrm{Q}$         & 29.0                          & 0\%       & $\leq$ 0\%        &4.63        & 0.32  \\[1.5pt]
$[$10x5$]_3^\mathrm{TOQ}$       & 6.2                           & 0\%       & $\leq$ 0\%        &11.71       & 0.82  \\[1.5pt]
\midrule
$[$10x5$]_4^\mathrm{TO}$        & 4.53                          & 2.1\%     & N/A               &8.58        & 0.48  \\[1.5pt]
$[$10x5$]_4^\mathrm{Q}$         & 29.0                          & 0\%       & $\leq$ 0\%        &9.34        & 0.52  \\[1.5pt]
$[$10x5$]_4^\mathrm{TOQ}$       & 6.2                           & 0\%       & $\leq$ 0\%        &24.42        & 1.37 \\[1.5pt]
\midrule
$[$15x15$]_3^\mathrm{TO}$       & 25.2                          & 10\%      & N/A               &36.15       & 0.33  \\[1.5pt]
$[$15x15$]_3^\mathrm{Q}$        & 36.66                         & 6\%       & $\leq$ 6\%        &36.21       & 0.34  \\[1.5pt]
$[$15x15$]_3^\mathrm{TOQ}$      & 31.72                         & 6\%       & $\leq$ 6\%        &86.32       & 0.81  \\[1.5pt]
\midrule
$[$15x15$]_4^\mathrm{TO}$       & 25.2                          & 10\%      & N/A               &75.66       & 0.58  \\[1.5pt]
$[$15x15$]_4^\mathrm{Q}$        & 37.83                         & 3\%       & $\leq$ 8.9\%      &73.97       & 0.56  \\[1.5pt]
$[$15x15$]_4^\mathrm{TOQ}$      & 29.86                         & 3\%       & $\leq$ 3.5\%      &186.52      & 1.43 \\[1.5pt]
\midrule
\end{tabular}\caption{Comparison between the TOQ, Q and TO policies computed approximating Problems~\eqref{eq:minTimeQua}-\eqref{eq:minTime}, respectively. In the table, the grid world $[$XxY$]_{i}^j$ is defined by XxY grid cells,  $i$ uncertain regions and the control policy $j\in\{\mathrm{TO}, \mathrm{Q}, \mathrm{TOQ}\}$.}\label{table:comparison}
\end{table}

% For all policies, we running $2^{n_u}$ roll-outs for all possible configurations of the $n_u$ uncertain regions, given that at test time a correct measurement is received by the agent. The realized cost and number of failures from these roll-outs are then weighted by the initial belief $\be$ to approximate the expected cost and probability of failure. 

\vspace{-0.5cm}
Table~\ref{table:comparison} shows the expected time to complete the task, the probability of failure, the upper-bound of the probability of failure\footnote{For the TOQ policy the probability of failure $(1 - \mathbb{P}^\pi[\boldsymbol{\omega} \models \phi] ) \leq 1 - J_0(s,\be)$ where $J_0(s,\be)$ is defined in Lemma~\ref{th:appValFun}.}, the total time to approximate the value function and the backup time required to approximate the value function at a belief point $\be \in \mathcal{D}_b$.
The TO and Q policies are approximated using a standard point-based backup update and the TOQ policy is computed using the $\texttt{backup}$ function from Algorithm~\ref{algo:backup}.
The TO policy from Problem~\eqref{eq:minTime} minimizes the expected time to complete the task but, as a result, it incurs in the highest probability of failure. On the other hand, the proposed TOQ policy has a probability of failure equal to the Q policy, which is computed maximizing the probability of satisfying the specifications.  
Therefore, the proposed strategy is able to minimize the expected time to complete the task, without compromising the probability of mission success. 
Notice that as a trade-off the computational burden of synthesizing the proposed TOQ policy is higher compared to the one needed to synthesize the Q and TO policies. 
This result is expected as we are approximating the solution to Problem~\eqref{eq:minTimeQua} using a pair of vectors; whereas, the point-based strategy used to approximate Problems~\eqref{eq:maxSpecs}-\eqref{eq:minTime} maintains a single support vector per belief point. 
Finally, we underline that the backup time shown in Table~\ref{table:comparison} is associated with the computation of the support vectors at a discrete belief point $\be\in\mathcal{D}_b$. Thus, it mostly depends on the dimension of the belief space, which grows exponentially with the number of uncertain regions. 
Clearly, the total time needed to synthesize the control policy depends also on the grid size and number of
belief points, as the $\texttt{backup}$ update from Algorithm~\ref{algo:valFunUpdate} is used repeatedly to approximate the value function. Indeed, when parallel computing is not available, the total computational cost scales linearly with the number of observable states $|\mathcal{S}|$ and discrete belief points~$|\mathcal{D}_b|$.

\subsection{Navigation Task}
In this section, we use the proposed time-optimal quantitative policy~\eqref{eq:minTimeQua} as high-level decision maker for the navigation problem shown in Figure~\ref{fig:simEnv}, where a Segway has to explore a partially known environment to locate science samples that may be located in the goal regions $\mathcal{G}_i$. The specification $\psi = \neg \texttt{collision} U (( \texttt{Goal}_1 \land \texttt{sample}_1) \lor ( \texttt{Goal}_2 \land \texttt{sample}_2) )$, where the atomic proposition $\texttt{sample}_i$ is \texttt{true} if the region $\mathcal{G}_i$ contains a science sample and the atomic proposition $\texttt{Goal}_i$ is \texttt{true} if the Segway is in a goal cell $\mathcal{G}_i$. We implemented a hierarchical controller, where the the proposed time-optimal quantitative policy~\eqref{eq:minTimeQua} computes high-level commands and a model predictive controller~\cite{borrelli2017predictive} is used to compute low-level inputs.
The high-level commands are move North, South, East and West and they are used to compute the cell where the Segway should move next. Then,
the low-level control problem is solved as a standard regulation problem~\cite{borrelli2017predictive}, where the goal is to steer the Segway to the center of the goal cell. When a transition from cell $i$ to cell $j$ occurs, we update the belief about the environment and the observable state of the MOMDP, which represents the cell where the Segway is located.
% The goal of the navigation problem is to find a science sample which may be located in one of the two goal regions $\mathcal{G}_1$ and $\mathcal{G}_2$, i.e., $\psi = \neg \texttt{collision} U \texttt{sample}$, where the atomic proposition \texttt{sample} is \texttt{true} if the robot is in a goal cell $\mathcal{G}_i$ which contains a science sample. While searching for the science sample, the Segway has to explore the uncertain regions $\mathcal{R}_1$, $\mathcal{R}_2$ and $\mathcal{R}_3$,  which may be traversable with some probability. 
The accuracy of the environment observations decays exponentially as a function of the distance between the Segway and the measured region. In particular, for the binary variable $r^{(i)}\in\{0,1\}$, which represents the traversability of the region $\mathcal{R}_i$, we  receive a measurement $z_r^{(i)}$ which is accurate with the following probability:
\begin{equation*}
\begin{aligned}
P(z_r^{(i)}&=1|r^{(i)}=1, s)\\&=\begin{cases}
    1 & \mbox{if }d(s, \mathcal{R}_i)\leq 1, \\
    0.5+0.3e^{-(d(s, \mathcal{R}_i)-2)/2.5} & \mbox{otherwise}, \\
\end{cases}
\end{aligned}
\end{equation*}
where $d(s, \mathcal{R}_i)$ represents the Manhattan distance between the Segway and region $\mathcal{R}_i$. Similarly, we define the binary variable $g^{(i)}\in\{0,1\}$, which equals to one when region $\mathcal{G}_i$ contains a science sample and zero otherwise, and we receive an observation $z_g^{(i)}$  which has the following accuracy:
\begin{equation*}
\begin{aligned}
P(z_g^{(i)}=1|&g^{(i)}=1, s)\\&=\begin{cases}
    1 & \mbox{if }d(s, \mathcal{R}_i)=0, \\
    0.5+0.25e^{-d(s, \mathcal{R}_i)/1.5} & \mbox{otherwise}. \\
\end{cases}
\end{aligned}
\end{equation*}
% We constructed a finite state automaton with state $\mathcal{Q} = \{ \texttt{free state}, \texttt{collision state}, \texttt{terminal state} \}$ that is initialized at the \texttt{free state}. A transition to a \texttt{collision state} occurs when the Segway is in an uncertain region that is not traversavle and a transition to the \texttt{terminal state} when the goal state that contains a science sample is reached. Afterwards, we constructed the product MOMDP and we removed that states which are not reachable from the starting state to reduce the computational~load.
\begin{figure}[t!]
    \centering
	\includegraphics[width= 1.0\columnwidth]{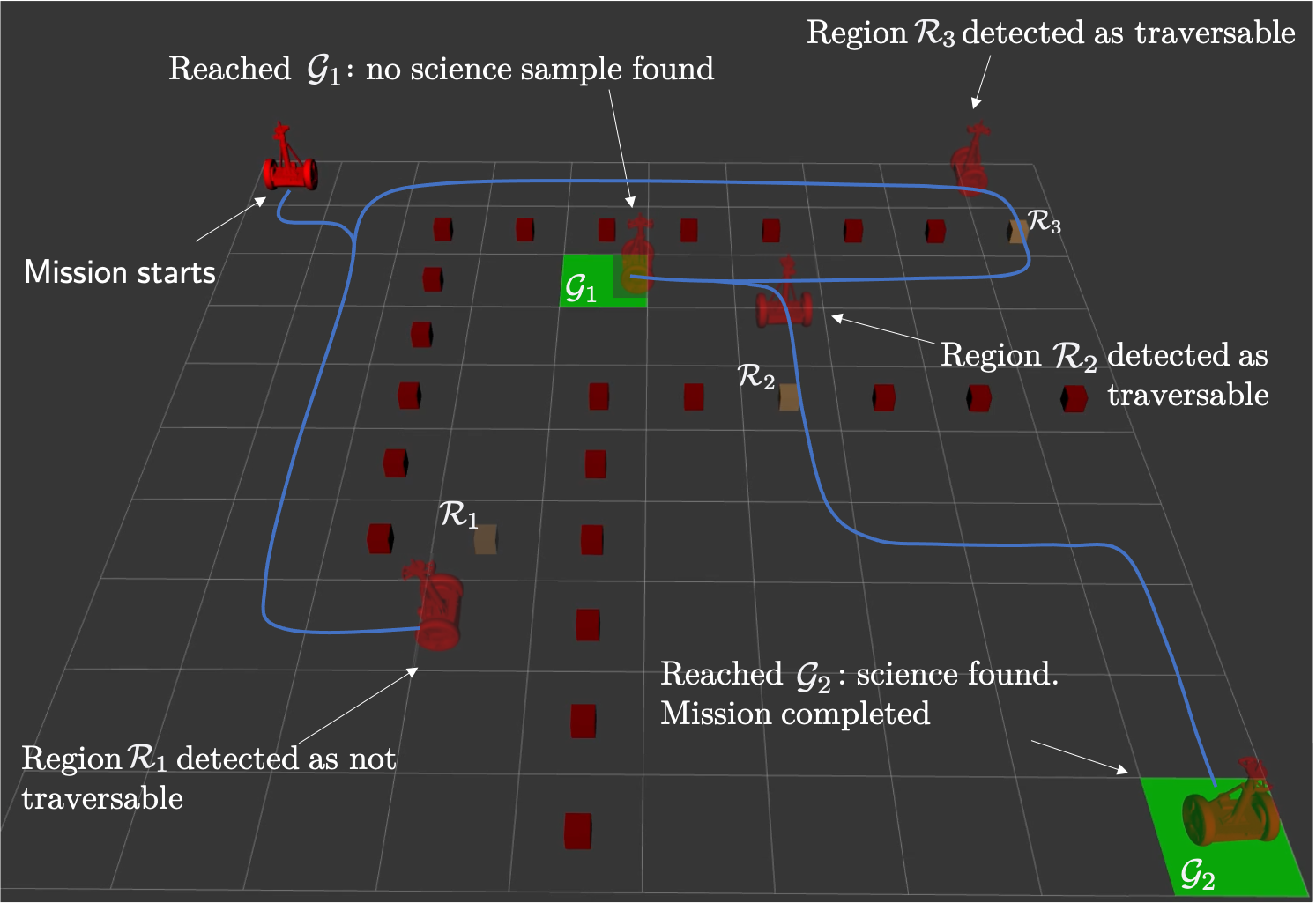}
    \caption{Evolution of the Segway (blue) in the high-fidelity simulator. The TOQ policy decides to first explore region $\mathcal{R}_1$, which in this example is not traversable. Afterwards, the controller explores regions $\mathcal{R}_2$, $\mathcal{G}_1$ and $\mathcal{R}_3$. Finally, the Segway reaches region $\mathcal{G}_2$, which in this example contains the science sample.}\label{fig:simEnv}
\end{figure}

\begin{figure}[t]
    \centering
	\includegraphics[width= 1.0\columnwidth]{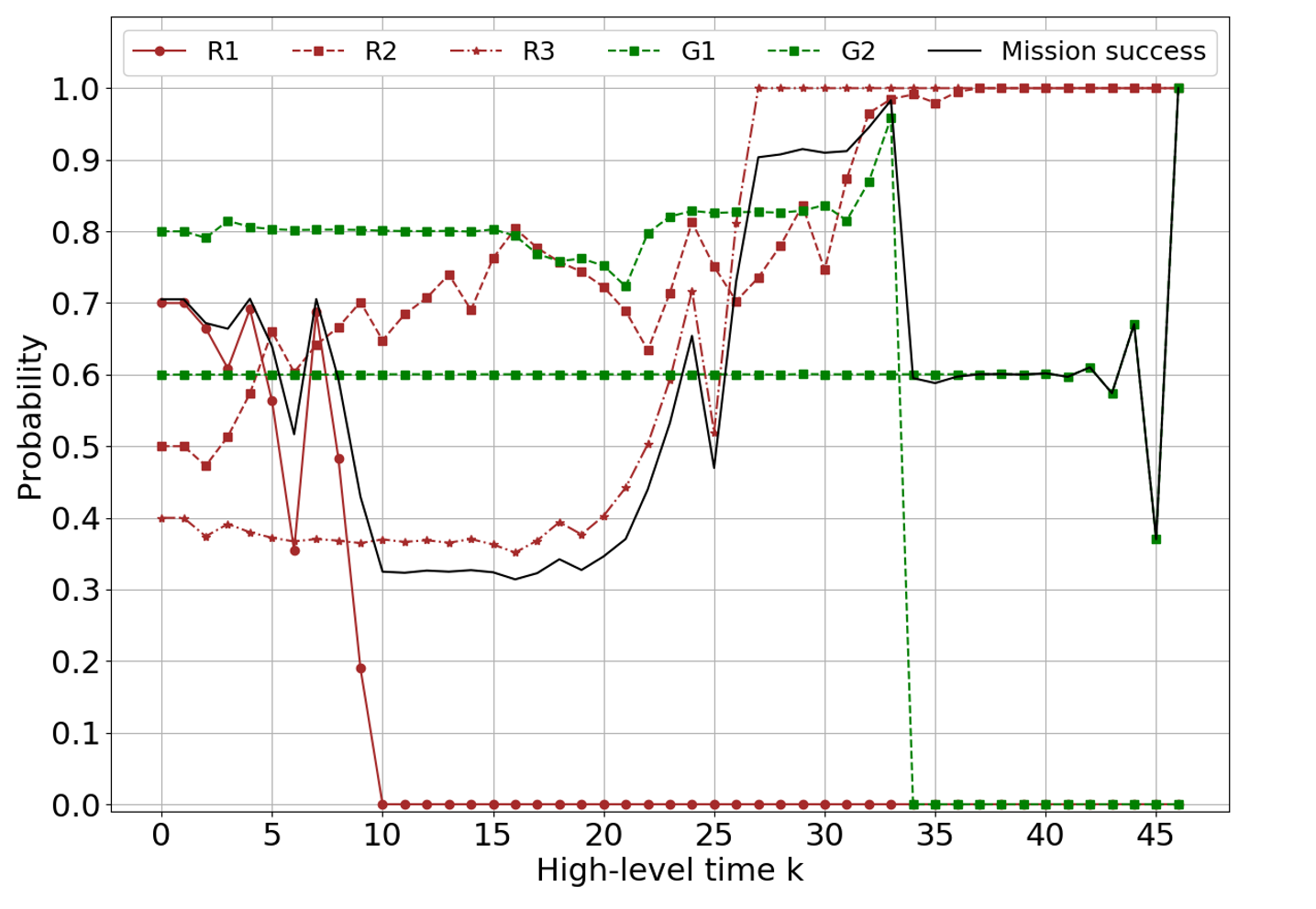}
    \caption{Probability of satisfying the specification. The figure shows also the evolution of the belief for the uncertain and goal regions.}\label{fig:probSpec}
\end{figure}

Figure~\ref{fig:simEnv} shows the closed-loop trajectory of the Segway. At the beginning of the simulation, the probability that regions $\mathcal{R}_1$, $\mathcal{R}_2$, and $\mathcal{R}_3$, may be traversable is $0.7$, $0.5$, and $0.4$. Furthermore, the probability that regions $\mathcal{G}_1$ and $\mathcal{G}_2$ contain the science sample is $0.8$ and $0.6$, respectively. The controller first explores region $\mathcal{R}_1$, which has the highest probability of being traversable. However, in this example region $\mathcal{R}_1$ is not traversable and therefore the Segway steers to region $\mathcal{R}_3$. 
As shown in Figure~\ref{fig:probSpec}, the environment observations are used to update environment beliefs and the probability of mission success, which represents the probability of satisfying the mission specifications. Notice that, when the Segway detects that region $\mathcal{G}_1$ does not contain a science sample at the high-level time $k=34$, the probability of mission success drops, as shown in Figure~\ref{fig:probSpec}. Afterwards, the controller explores region $\mathcal{R}_3$ and steers the Segway to region $\mathcal{G}_2$, which contains a science sample. Finally, we notice that for all high-level time steps $k \geq 37$ the controller is uncertain only about the state of region $\mathcal{G}_2$, therefore the probability of mission success overlaps with the probability that region $\mathcal{G}_2$ contains a science sample.

\section{Conclusions}
In this work, we studied time-optimal quantitative problems for MOMDPs. First, we presented a dynamic programming update to compute the value function of time-optimal quantitative problems. Afterwards, we leveraged the piecewise-affine nature of the optimal value function to define a point-based approximation strategy, which allows us to compute a lower bound of the probability of satisfying the specifications.
Finally, we compared the proposed strategy with time-optimal and quantitative policies.
% in three grid worlds. We showed that the proposed strategy, which is tailored to problem with mixed observability, can be used for navigation tasks where the state of the environment is partially observable.

% Future work will investigate strategies to further decompose the belief space to gain computationally tractability when environment states are observable in only specific parts of the environment and risk-sensitive decision making in partially observable environments~\cite{ahmadi2020risk}.

\renewcommand{\baselinestretch}{0.93}
\bibliographystyle{IEEEtran}
\bibliography{refs}

\end{document}